\documentclass[runningheads]{llncs}
\usepackage[T1]{fontenc}
\usepackage{graphicx}
\usepackage{amsmath,amssymb}\usepackage[countsame,nosections]{coqtheorem}
\usepackage{mathpartir}
\usepackage{bbm}
\usepackage{macros}
\usepackage{todonotes}
\usepackage{tikz-cd}

\usepackage{cleveref}
\usepackage[countsame,nosections]{coqtheorem}

\setBaseUrl{{https://uds-psl.github.io/coq-synthetic-computability/oracle-computability/SyntheticComputability.}}

\Crefname{corollary}{Corollary}{Corollaries}
\Crefname{theorem}{Theorem}{Theorems}
\Crefname{lemma}{Lemma}{Lemmas}
\Crefname{corollaryAuxCoq}{Corollary}{Corollaries}
\Crefname{theoremAuxCoq}{Theorem}{Theorems}
\Crefname{lemmaAuxCoq}{Lemma}{Lemmas}
\Crefname{corollaryAux}{Corollary}{Corollaries}
\Crefname{theoremAux}{Theorem}{Theorems}
\Crefname{lemmaAux}{Lemma}{Lemmas}

\begin{document}

\title{Oracle Computability and Turing Reducibility\\in the Calculus of Inductive Constructions\thanks{Yannick Forster received funding from the European Union’s Horizon 2020 research and innovation programme under the Marie Skłodowska-Curie grant agreement No.\ 101024493. Dominik Kirst is supported by a Minerva Fellowship of the Minerva Stiftung Gesellschaft fuer die Forschung mbH.}}

\titlerunning{Oracle Computability and Turing Reducibility in CIC}

\author{}
\institute{}
\author{Yannick Forster\inst{1}\orcidID{0000-0002-8676-9819} \and\\
  Dominik Kirst\inst{2,3}\orcidID{0000-0003-4126-6975} \and\\
  Niklas Mück\inst{3}\orcidID{0009-0006-9622-0762}}
\authorrunning{Y.\ Forster, D.\ Kirst, N.\ Mück}
\institute{Inria, LS2N, Université Nantes, France\\
  \email{yannick.forster@inria.fr}
  \and
  Ben-Gurion University of the Negev, Beer-Sheva, Israel\\
  \email{kirst@cs.bgu.ac.il}
  \and
  Saarland University and MPI-SWS, Saarland Informatics Campus, Saarbrücken, Germany\\
  \email{s8nimuec@stud.uni-saarland.de}
}

\maketitle              %
\setcounter{footnote}{0}
\begin{abstract}
  We develop synthetic notions of oracle computability and Turing reducibility in the Calculus of Inductive Constructions (CIC), the constructive type theory underlying the Coq proof assistant.
  As usual in synthetic approaches,
  we employ a definition of oracle computations based on
  meta-level functions rather than object-level models of computation,
  relying on the fact that in constructive systems such as CIC all definable functions are computable by construction.
  Such an approach lends itself well to machine-checked proofs, which we carry out in Coq.

  There is a tension in finding a good synthetic rendering of the higher-order notion of oracle computability.
  On the one hand, it has to be informative enough to prove central results, ensuring that all notions are faithfully captured.
  On the other hand, it has to be restricted enough to benefit from axioms for synthetic computability, which usually concern first-order objects.
  Drawing inspiration from a definition by Andrej Bauer based on continuous functions in the effective topos, we use a notion of sequential continuity to characterise valid oracle computations.

  As main technical results, we show 
  that Turing reducibility forms an upper semilattice, transports decidability, and is strictly more expressive than truth-table reducibility,
  and prove that whenever both a predicate $p$ and its complement are semi-decidable relative to an oracle $q$, then $p$ Turing-reduces to $q$.
\end{abstract}

\section{Introduction}

In recent years, synthetic computability theory~\cite{richman1983church,bridges1987varieties,BauerSyntCT,bauer2017fixed} has gained increasing attention in the fields of constructive mathematics and interactive theorem proving~\cite{forster2019synthetic,forster2020churchs,forster2022parametric,kirst2023godel,swan2019church,swanoracle}.
In contrast to the usual analytic approach based on describing the functions considered computable by means of a model like Turing machines, $\mu$-recursive functions, or the $\lambda$-calculus, the synthetic approach exploits that in a constructive setting no non-computable functions can be defined in the first place, making a later description of the computable fragment obsolete.
This idea enables much more compact definitions and proofs, for instance decidability of sets over $\nat$ can be expressed by equivalence to functions $f\of{\nat\to\bool}$ without any further computability requirement regarding $f$, simplifying a formal mathematical development and being the only approach enabling a feasible mechanisation using a proof assistant.
Concerning the logical foundations of programming, in constructive type theories such as the Calculus of Inductive Constructions (CIC)~\cite{coquand1986calculus,paulin1993inductive,paulin2015introduction} underlying the Coq proof assistant~\cite{Coq}, synthetic computability is especially natural: as CIC embodies a dependently-typed functional programming language, every definable function conveys its own executable implementation.

Despite the fruitful use of the synthetic approach to describe basic concepts in computability theory, the characterisation of oracle computations in general (i.e.\ algorithms relative to some potentially non-computable subroutine) and Turing reductions in particular (i.e.\ decision procedures relative to some oracle giving answer to a potentially non-decidable problem) proves more complicated.
First, a Turing reduction cannot naively be described by a transformation of computable decision procedures $\nat\to\bool$ as this would rule out the intended application to oracles for problems that can be proved undecidable using usual axioms of synthetic computability such as Church's thesis (CT).
Secondly, when instead characterising Turing reductions by transformations of possibly non-computable decision procedures represented as binary relations $\nat\to\bool\to\Prop$, one has to ensure that computability is preserved in the sense that computable oracles induce computable reductions in order to enable intended properties like the transport of (un-)decidability.
Thirdly, to rule out exotic reductions whose behaviour on non-computable oracles differs substantially from their action on computable oracles, one needs to impose a form of continuity.

The possible formulations of continuity of functionals on partial spaces such as $\nat\to\bool\to\Prop$ are numerous: Bauer~\cite{bauer2020Wisc}, who gave the first synthetic definition of oracle computability we draw our inspiration from, employs the order-theoretic variant of functionals preserving suprema in directed countable partial orders.
Forster~\cite{forster2022thesis} describes a reformulation to CIC in joint work with Kirst, using a modified variant of modulus-continuity where every terminating oracle computation provides classical information about the information accessed from the oracle.
Another preliminary suggestion due to Forster and Kirst~\cite{types22abstract} uses a more constructive formulation of modulus-continuity, allowing to establish Post's theorem connecting the arithmetical hierarchy with Turing degrees~\cite{types22kleenepost}.
However, their proof assumes an enumeration of all (higher-order) oracle computations which seems not to follow from CT, therefore leaving the consistency status of the assumption unclear.

As a remedy to this situation, we propose an alternative synthetic characterisation of oracle computability based on a stricter notion of sequential continuity, loosely following van Oosten~\cite{vanOosten2011-VANPCA-2}.
While this concept naturally describes the functionals considered computable by emphasising the sequence of computation steps interleaved with oracle interactions, it immediately yields the desired enumeration from CT by reducing higher-order functionals on partial spaces to partial first-order functions on mere data types.
Concretely, 
in this paper we develop the theory of oracle computability as far as possible without any axioms for synthetic computability: we show %
that Turing reducibility forms an upper semilattice, transports decidability, and is strictly more expressive than truth-table reducibility,
and prove that whenever both a predicate $p$ and its complement are semi-decidable relative to an oracle $q$, then $p$ Turing-reduces to $q$.\footnote{The non-relativised form of the latter statement also appears under the name of ``Post's theorem'' in the literature~\cite{troelstra1988constructivism}, not to be confused with the mentioned theorem regarding the arithmetical hierarchy, see the explanation in \Cref{sec:PT}.}
All results are mechanised in Coq, both to showcase the feasibility of the synthetic approach for machine-checked mathematics and as base for future related mechanisation projects.

For easy accessibility, the Coq development\footnote{\href{https://github.com/uds-psl/coq-synthetic-computability/tree/code-paper-oracle-computability}{\texttt{https://github.com/uds-psl/coq-synthetic-computability/}} \href{https://github.com/uds-psl/coq-synthetic-computability/tree/code-paper-oracle-computability}{\texttt{tree/code-paper-oracle-computability}}} is seamlessly integrated with the text presentation: every formal statement in the PDF version of this paper is hyperlinked with HTML documentation of the Coq code.
To further improve fluid readability, we introduce most concepts and notations in passing, but hyperlink most definitions in the PDF with the glossary in \Cref{sec:glossary}.
\vspace{-0.8\baselineskip}

\paragraph{Contribution}
We give a definition of synthetic oracle computability in constructive type theory and derive notions of Turing reducibility and relative semi-decidability.
We establish basic properties of all notions,
most notably that Turing reducibility forms an upper semi-lattice,
transports decidability if and only if Markov's principle holds,
and is strictly more general than truth-table reducibility.
We conclude by a proof of Post's theorem relating decidability with semi-decidability of a set and its complement.
\vspace{-0.8\baselineskip}

\paragraph{Outline}
We begin by introducing the central notion of synthetic oracle computability in \Cref{sec:oracle}, employed in \Cref{sec:turing} to derive synthetic notions of Turing reducibility and oracle semi-decidability.
Before we discuss their respective properties (\Cref{sec:prop_sdec,sec:prop_turing}) and show that Turing reducibility is strictly weaker than a previous synthetic rendering of truth-table reducibility (\Cref{sec:tt}), we develop the basic theory of synthetic oracle computations by establishing their closure properties (\Cref{sec:closure}) and by capturing their computational behaviour (\Cref{sec:functions}).
Some of these closure properties rely on a rather technical alternative characterisation of oracle computability described in \Cref{sec:forms}, which will also be used to establish the main result relating oracle semi-decidability with Turing reducibility discussed in \Cref{sec:PT}.
We conclude in \Cref{sec:conclusion} with remarks on the Coq formalisation as well as future and related work.

\section{Synthetic Oracle Computability}
\label{sec:oracle}

The central notion of this paper is the synthetic definition of oracle computability.
Historically, oracle computability was introduced as an extension of Turing machines in Turing's PhD thesis~\cite{turing1939systems},
but popularised by Post~\cite{post1944recursively}.
Various analytic definitions of oracle computability exist, all having in common that computations can ask questions and retrieve answers from an oracle.

For our synthetic definition, we specify concretely when a higher-order functional $F\of{(Q \to A \to \Prop) \to (I \to O \to \Prop)}$
is considered (oracle-)computable.
Such a functional takes as input a possibly non-total binary relation $R\of{Q\to A\to \Prop}$, an oracle relating questions $q\of{Q}$ to answers $a\of{A}$, and yields a computation relating inputs $i\of{I}$ to outputs $o\of{O}$.
For special cases like Turing reductions, we will instantiate $Q,I:=\nat$ and $A,O:=\bool$.
Note that we do not require oracles $R$ to be deterministic, but if they are, then so are the resulting relations $F R$ (cf.\ \Cref{coq:interrogation_output_det}).

We define oracle computability by observing that a terminating computation with oracles has a sequential form:
in any step of the sequence, the oracle computation can ask a question to the oracle, return an output, or diverge.
Informally, we can enforce such sequential behaviour by requiring that every terminating computation $F R\,i\,o$ can be described by (finite, possibly empty) lists $\qs\of{\List Q}$ and $\ans\of{\List A}$ such that from the input $i$ the output $o$ is eventually obtained after a finite sequence of steps, during which the questions in $\qs$ are asked to the oracle one-by-one, yielding corresponding answers in $\ans$.
This computational data can be captured by a partial\footnote{There are many ways how semi-decidable partial values can be represented in CIC, for instance via step-indexing. Since the actual implementation does not matter, we abstract over any representation providing the necessary operations, see \Cref{def:partial}.} function of type $\ofbox{I \to \List{A} \pto \sumt Q O}$, called the (computation) tree of $F$, that on some input and list of previous answers either returns the next question to the oracle, returns the final output, or diverges.

So more formally, we call $F\of{(Q \to A \to \Prop) \to (I \to O \to \Prop)}$ 
an (oracle-)computable functional if
there is a  tree $\tau \of{I \to \List{A} \pto \sumt Q O}$ such that
\[ \forall R\,i\,o.\;F R\,i\,o ~\leftrightarrow~ \exists \qs \; \ans. ~ \interrogate {\tau i} R \qs \ans \,\land\, \tau\,i\,\ans \hasvalue \inr o  \]
with  the interrogation relation $\sigma ; R \vdash \qs ; \ans$ being defined inductively by
\label{intro:interrogate}
\begin{mathpar}
  \infer{~}{\interrogate \sigma R {[]} {[]}}

  \infer{\interrogate \sigma R \qs \ans \and \sigma \ans \hasvalue \inl q \and R q a}
    {\interrogate \sigma R {\qs \app [q]} {\ans \app [a]}}
\end{mathpar}
where $\List A$ is the type of lists over $a$\label{intro:list},
$l\app l'$ is list concatenation\label{intro:app}, where we use the suggestive shorthands $\inl q$ and $\inr o$ for the respective injections into the sum type $\sumt Q O$\label{intro:sum}, and where $\sigma\of{\List{A} \pto \sumt Q O}$ denotes a tree at a fixed input $i$.

To provide some further intuition and visualise the usage of the word ``tree'', we discuss the following example functional in more detail:
\begin{align*}
	F ~&:~ (\nat\to\bool\to\Prop)\to(\nat\to\bool\to\Prop)\\
	FR\,i\,o~&:=~o=\btrue \land \forall q < i.\, R\,q\,\btrue
\end{align*}

Intuitively, the functional can be computed by asking all questions $q$ for $q < i$ to the oracle.
If the oracle does not return any value, $F$ does not return a value.
If the oracle returns $\bfalse$ somewhere, $F$ also does not return a value -- i.e.\ runs forever.
If the oracle indeed returns $\btrue$ for all $q < i$, $F$ returns $\btrue$.

In the case of $i=3$, this process may be depicted by

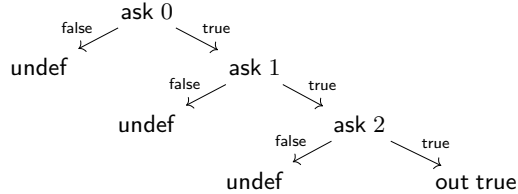
\begin{center}\footnotesize
\begin{tikzcd}[sep=small]
& {\ask 0} \arrow[ld, "\bfalse"'] \arrow[rd, "\btrue"] &                                     &                                     &    \\
{\partundef} &                                     & {\ask 1} \arrow[ld, "\bfalse"'] \arrow[rd, "\btrue"] &                                     &    \\
& {\partundef}                                  &                                     & {\ask 2} \arrow[ld, "\bfalse"'] \arrow[rd, "\btrue"] &    \\
&                                     & {\partundef}                                  &                                     & {\out \btrue}
\end{tikzcd}
\end{center}

\noindent
where the paths along labelled edges represent the possible answer lists $\ans$ while the nodes represent the corresponding actions of the computation: the paths along inner nodes denote the question lists $\qs$ and the leafs the output behaviour.
Note that $\ret\of{X\pto X}$ is the return of partial functions, turning a value into an always defined partial value, while $\partundef$ denotes the diverging partial value.
Formally, a tree $\tau\of{\nat \to \List{\bool} \pto \sumt \nat \bool}$ computing $F$ can be defined by
$$\tau\,i\,\ans~:=~
\begin{cases}
\partundef&\textnormal{if }\bfalse \in \ans\\
\ret(\ask\, |\ans|)&\textnormal{if }\bfalse \not\in \ans \land |\ans|< i\\
\ret(\out\,\btrue)&\textnormal{if }\bfalse \not\in \ans \land |\ans|\ge i
\end{cases}$$
where here and later on we use such function definitions by cases to represent (computable) pattern matching.

As usual in synthetic mathematics, the definition of a functional $F$ as being computable if it can be described by a tree is implicitly relying on the fact that all definable (partial) functions in CIC could also be shown computable in the analytic sense.
Describing oracle computations via trees in stages goes back to Kleene~\cite{Kleene_1959}, cf.\ also the book by Odifreddi~\cite{odifreddi1992classical}.
Our definition can be seen as a more explicit form of sequential continuity due to van Oosten~\cite{van_Oosten_1999,vanOosten2011-VANPCA-2},
or as a partial, extensional form of a dialogue tree due to Escardó~\cite{escardodialogue}.
Our definition allows us to re-prove the theorem
by Kleene~\cite{kleene1952introduction} and Davis~\cite{Davis1958-DAVCU}
that
computable functionals fulfill the more common definition of continuity with a modulus:

\setCoqFilename{TuringReducibility.OracleComputability}
\begin{lemma}[][cont_to_cont]
  Let $F$ be a computable functional.
  If $F R\,i\,o$, then there exists a list $\qs \of{\List Q}$,
  the so-called modulus of continuity,
  such that $\forall q \in \qs.\;\exists a.\;R q a$
  and 
  for all $R'$ with $\forall q \in \qs.\,\forall a.\;R q a \leftrightarrow R' q a$ we also have that $F R'\, i\, o$.
\end{lemma}
\begin{proof}
  Given $F R\,i\,o$ and $F$ computable by $\tau$ we have $\interrogate {\tau i} {R} {\qs} {\ans}$ and $\tau\,i\,\ans \hasvalue \inr o$.
  It suffices to prove both $\forall q \in \qs.\;\exists a.\;R q a$ and $\interrogate {\tau i} {R'} {\qs} {\ans}$ by induction on the given interrogation, which is trivial.
  \qed
\end{proof}

Nevertheless, our notion of computable functionals is strictly stronger than modulus-continuity as stated,
while we are unaware of a proof relating it to a version where the moduli are computed by a partial function.

\begin{lemma}[][counterex]
	There are modulus-continuous functionals that are not computable.
\end{lemma}

\begin{proof}
	Consider the functional $F\of{(\nat\to\bool\to\Prop)\to(I\to O\to\Prop)}$ defined by
	$$FRio ~:=~ \exists q.\,R\,q\,\btrue.$$
	Clearly, $F$ is modulus-continuous since from a terminating run $FRio$ we obtain $q$ with $R\,q\,\btrue$ and therefore can choose $\qs:=[q]$ as suitable modulus.

  However, suppose $\tau:I \to \List{\bool}\pto \sumt \nat O$ were a tree for $F$, then given some input $i$ we can inspect the result of $\tau\,i\,[]$ because
  $F\,R_\top\,i\,o$ holds for all $i$, $o$, and the full oracle $R_\top\,q\,a := \top$.
  However, the result cannot be $\out o$ for any output $o$, as this would yield $F R_\bot$ for the empty oracle $R_\bot \,q\,a := \bot$, violating the definition of $F$.
  Thus $\tau\,i\,[] \hasvalue \ask q_0$, conveying an initial question $q_0$ independent of the input oracle.
	But then employing the oracle $R_0$ defined by $R_0\,q_0\,a:=\bot$ and $R_0\,q\,a:=\top$ for all $q\not=q_0$ we certainly have $F\,R_0\,i\,o$ by definition but no interrogation $\interrogate {\tau i} {R_0} \qs \ans$ with $\tau\,i\,\ans \hasvalue \inr o$, as this would necessarily include an answer $a$ with $R_0\,q_0\,a$ as first step, contradicting the construction of $R_0$.
	\qed
\end{proof}

The advantage of using the stricter notion of sequential continuity over modulus-continuity is that by their reduction to trees, computable functionals are effectively turned into flat first-order functions on data types. 
Thus one directly obtains an enumeration of all oracle computations, as needed in most advanced scenarios, from an enumeration of first-order functions, which itself could be obtained by assuming usual axioms for synthetic computability.

\section{Turing Reducibility and Oracle Semi-Decidability}
\label{sec:turing}

Using our synthetic notion of oracle computability, we can directly derive synthetic formulations of two further central notions of computability theory:
Turing reducibility -- capturing when a predicate is decidable relative to a given predicate -- 
and oracle semi-decidability -- capturing when a predicate can be recognised relative to a given predicate.

To provide some intuition first, we recall that in the synthetic setting a predicate $p:X\to \Prop$ over some type $X$ is decidable if there is a function $f\of{X\to\bool}$
such that $\forall x.\;p x \leftrightarrow f x = \btrue$, i.e.\ $f$ acts as a decider of $p$.
This definition is standard in synthetic computability~\cite{BauerSyntCT,forster2019synthetic} and relies on the fact that constructively definable functions $f\of{X\to\bool}$ are computable.

To relativise the definition of a decider to an oracle, we first define the characteristic relation $\charrel p\of{X\to\bool\to\Prop}$ of a predicate $p\of{X\to\Prop}$ by\label{intro:charrel}
$$\charrel p~:=~\lambda x b.
\begin{cases}
px & \text{if } b = \btrue \\
\neg p x & \text{if } b = \bfalse.\\
\end{cases}$$
Employing $\charrel p$, we can now equivalently characterise a decider $f$ for $p$ by requiring that $\forall xb.\;\charrel p xb \leftrightarrow f x = b$.
Relativising this exact pattern, we then define Turing reducibility of a predicate $p \of{X\to\Prop}$ to $q\of{Y \to \Prop}$ by a computable functional $F$ transporting the characteristic relation of $q$ to the characteristic relation of $p$:
\label{intro:redT}
\[ p \redT q ~:=~ \exists F.\; F \textit{ is computable} \land \forall xb.\; \charrel p x b \leftrightarrow F \charrel q x b \]

Note that while we do not need to annotate a decider $f$ with a computability condition because we consider all first-order functions of type $\ofbox{\nat\to\nat}$ or $\ofbox{\nat\to\bool}$ as computable,
a Turing reduction is not first-order,
and thus needs to be enriched with a tree to rule out unwanted behaviour.
In fact, without this condition, we would obtain $p \redT q$ for every $p$ and $q$ by simply setting $F\,R := \charrel p$.

Next, regarding semi-decidability, a possible non-relativised synthetic definition is to require a partial function $f\of{X\pfun\mathbbm{1}}$
such that $\forall x.\;p x \leftrightarrow f x \hasvalue \star$, where $\mathbbm{1}$ is the inductive unit type with singular element $\star$\label{intro:unit}.
That is, the semi-decider $f$ terminates on elements of $p$ and diverges on the complement $\overline p$ of $p$ (cf.~\cite{forster2022thesis}).

Again relativising the same pattern, we say that $p\of{X\to\Prop}$ is (oracle-)semi-decidable relative to $q\of{Y \to \Prop}$
if there is a computable functional $F$ mapping relations $R\of{Y\to\bool\to\Prop}$ to
relations of type $\ofbox{X\to\mathbbm{1}\to\Prop}$ such that $F \charrel q$ accepts $p$:

\newcommand{\sdec}[2]{\ensuremath{\mathcal{S}_{#1}(#2)}}

\[ \sdec q p ~:=~ \exists F.\; F \textit{ is computable} \land \forall x.\; p x \leftrightarrow F \charrel q x \star \]

As in the case of Turing reductions, the computability condition of an oracle semi-decider is crucial: without the restriction, we would obtain $\sdec q p$ for every $p$ and $q$ by setting $F\,R\,x\,\star:=p\,x$.

While we defer developing the theory of synthetic Turing reducibility and oracle semi-decidability to later sections, we can already record here that the fact that decidability implies semi-decidability also holds in relativised form:

\setCoqFilename{TuringReducibility.OracleComputability}
\begin{lemma}[][Turing_to_sdec]
	If $p\redT q$ then $\sdec q p$ and $\sdec q {\overline p}$.
\end{lemma}
\begin{proof}
	Let $F$ witnesses $p\redT q$, then $F'\,R\,x\,\star := F\,R\,x\,\btrue$ witnesses $\sdec q p$.
	In particular, if $\tau\of{X\to \List{\bool}\pto\nat +\bool}$ computes $F$, then $\tau'\of{X\to \List{\bool}\pto\nat +\mathbbm{1}}$, constructed by running $\tau$ and returning $\inr \star$ whenever $\tau$ returns $\inr \btrue$, computes $F'$.
	The proof of $\sdec q {\overline p}$ is analogous, simply using $\bfalse$ in place of $\btrue$.
	\qed
\end{proof}

\section{Closure Properties of Oracle Computations}
\label{sec:closure}

In this section we collect some examples of computable functionals and show how they can be composed, yielding a helpful abstraction for later computability proofs without need for constructing concrete computation trees.
Note that the last statements of this section depend on a rather technical intermediate construction using a more flexible form of interrogations.
We refer to the Coq code and to \Cref{sec:forms}, where we will also deliver the proofs left out.

First, we show that composition with a transformation of inputs preserves computability and that all partial functions are computable, ignoring the the input oracle.
The latter also implies that total, constant, and everywhere undefined functions are computable.

\setCoqFilename{TuringReducibility.OracleComputability}
\begin{lemma}[][computable_precompose]\label{lem:allaboutcomp}
  The following functionals mapping relations $R \of{Q \to A \to \Prop}$ to relations
  of type $\ofbox{I \to O \to \Prop}$ are computable:
  \begin{enumerate}
  \item $\lambda R\,i\,o.\;F R\,(g i)\,o$  for $g \of{I \to I'}$ and computable $F\of{(Q \to A \to \Prop )\to (I' \to O \to \Prop)}$,
  \item $\lambda R\,i\,o.\;f i \hasvalue o$ given $f \of{I \pfun O}$,
  \item $\lambda R\,i\,o.\;f i = o$ given $f \of{I \to O}$,
  \item $\lambda R\,i\,o.\;o = v$ given $v \of O$,
  \item $\lambda R\,i\,o.\;\bot$.
  \end{enumerate}
\end{lemma}
\begin{proof}
  For 1, let $\tau$ compute $F$ and define $\tau'\,i\,l := \tau\, (g i)\, l$.
  For 2, define $\tau'\,i\,l := f i \bind \lambda o.\; \ret (\out o)$,
  where $\bind$ is the bind operation of partial functions.
  All others follow by using (2).
  \qed
\end{proof}

Next, if $Q=I$ and $A=O$, then the identity functional is computable:

\setCoqFilename{TuringReducibility.OracleComputability}
\begin{lemma}[][computable_id]
  The functional mapping $R \of{Q \to A \to \Prop}$ to $R$ itself is computable.
\end{lemma}
\begin{proof}
  Define
  \[ \tau\,q\,l :=
    \begin{cases}
      \ret (\inl q) & \text{if } l = [] \\
      \ret (\inr a) & \text{if } l = (q,a) :: l'. \tag*{\qed}
    \end{cases}
  \]
\end{proof}

Moreover, given two functionals and a boolean test on inputs, the process calling either of the two depending on the test outcome is computable:

\begin{lemma}[][computable_if]
  Let $F_1$ and $F_2$ both map relations $R\of{Q\to A \to \Prop}$
  to relations of type $\ofbox{I \to O \to \Prop}$ and
  $f\of{I \to \bool}$.
  Then $F$ mapping $R$ to
  the following relation of type $\ofbox{I \to O \to \Prop}$
  is computable:
  \[\lambda io.\begin{cases}
  F_1\; R\; i\; o & \text{if } fi = \btrue \\
  F_2\; R\; i\; o & \text{if } fi = \bfalse 
  \end{cases}
  \]
\end{lemma}
\begin{proof}
  Let $\tau_1$ and $\tau_2$ compute $F_1$ and $F_2$ respectively
  and define \[ \tau\,i\,l := \begin{cases}
  \tau_1\;i\;l & \text{if } fi = \btrue \\
  \tau_2\;i\;l & \text{if } fi = \bfalse. \tag*{\qed}
  \end{cases}\]
\end{proof}

Taken together, the previous three lemmas yield computability proofs for functionals consisting of simple operations like calling functions, taking indentities, and branching over conditionals.
The next three lemmas extend to partial binding, function composition, and linear search, so in total we obtain an abstraction layer accommodating computability proofs for the most common ingredients of algorithms.
As mentioned before, we just state the last three lemmas without proof here and refer to the Coq development and \Cref{sec:forms} for full detail.

\begin{lemma}[][computable_bind]
  Let $F_1$ map relations $R \of{Q\to A \to \Prop}$ to relations of type $\ofbox{I \to O' \to \Prop}$,
  $F_2$ map relations $R \of{Q\to A \to \Prop}$ to relations of type $\ofbox{(I \times O') \to O \to \Prop}$,
  and both be computable.
  Then $F$ mapping $R \of{Q\to A \to \Prop}$ to $\lambda i o.\,\exists o' \of O'.\,F_1\,R\,i\,o' \land F_2\,R\,(i,o')\,o$
  of type $\ofbox{I \to O \to \Prop}$
  is computable.
\end{lemma}

\begin{lemma}[][computable_comp]
  Let $F_1$ map relations $R \of{Q\to A \to \Prop}$ to relations $\ofbox{X \to Y \to \Prop}$,
  $F_2$ map relations $R \of{X\to Y \to \Prop}$ to relations $\ofbox{I \to O \to \Prop}$,
  and both be computable.
  Then $F$ mapping $R \of{Q\to A \to \Prop}$ to $\lambda i o.\;F_2\;(F_1 R)\;i\;o$
  of type $\ofbox{I \to O \to \Prop}$
  is computable.
\end{lemma}

\begin{lemma}[][computable_search]
  The functional mapping $R\of{(I \times \nat) \to\bool\to\Prop}$ to
  the following relation of type $\ofbox{I \to \nat \to \Prop}$ is computable:
  $\lambda i n\; R\; (i,n)\; \btrue \land \forall m < n.\; R\; (i, m)\; \bfalse$.
\end{lemma}

\section{Computational Cores of Oracle Computations}
\label{sec:functions}

In this section,
we prove that if $F$ maps $R\of{Q \to A \to \Prop}$ to a relation $\ofbox{I \to O \to \Prop}$
and $F$ is computable,
then there is a higher-order function $f \of{ (Q \pfun A) \to (I \pfun O)}$ such that
for any $r \of{Q \pfun A}$ with graph $R$,
the graph of $f r$ agrees with $F R$.
This means that every computable functional possesses an explicit computational core, mapping (partially) computable input to (partially) computable output, needed for instance to justify that decidability is transported backwards along Turing reductions (\Cref{coq:partial_decidable}).

In preparation, the following two lemmas state simple properties of interrogations regarding concatenation and determinacy.
Given $\sigma \of {\List{A} \pfun \sumt Q O}$ and $l \of{\List{A}}$ we write
$\sigma \at l$ for the sub-tree of $\sigma$ starting at path $l$, i.e.\ for the tree $\lambda l'.\;\sigma(l \app l')$.\label{intro:subtree}

\begin{lemma}[][noqinterrogation_app]\label{lem:noqinterrogation_app}
	We have interrogations $\interrogate \sigma R {\qs_1} {\ans_1}$ and $\interrogate {\sigma \at \ans_1} R {\qs_2} {\ans_2}$
	if and only if
	$|\qs_2|=|\ans_2|$ and $\interrogate \sigma R {\qs_1 \app \qs_2} {\ans_1 \app \ans_2}$.
\end{lemma}

\begin{lemma}[][interrogation_output_det]
	Let $R$ be functional and
	$\interrogate \sigma R {\qs_1} {\ans_1}$
	as well as
	$\interrogate \sigma R {\qs_2} {\ans_2}$.
	Then if $| \qs_1 | \leq |\qs_2|$,
	then $\qs_1$ is a prefix of $\qs_2$ and $\ans_1$ is a prefix of $\ans_2$.
\end{lemma}

Now conveying the main idea, we first define an evaluation function $\delta\; \sigma\; f: \nat \pfun \sumt Q O$ which evaluates
$\sigma \of{\List{A} \pfun \sumt Q O}$
on $f \of{Q \pfun A}$
for at most $n$ questions.
\begin{mathpar}
  \delta\; \sigma\; f\; n := \sigma [\,] \bind \lambda x.\;
  \begin{cases}
    \ret (\inr o) & \text{if } x = \inr o \\
    \ret (\inl q) & \text{if } x = \inl q, n = 0 \\
    f q \bind \lambda a.\;\delta\; (\sigma \at [a])\;f\;n' & \text{if } x = \inl q, n = \succN n'. \\
  \end{cases}
\end{mathpar}

The intuition is that $\delta$ always reads the initial node of the tree $\sigma$ by evaluating~$\sigma [\,]$.
If $\sigma [\,] \hasvalue \inr o$, then $\delta$ returns this output.
Otherwise, if $\sigma [\,] \hasvalue \inl q$ and $\delta$ has to evaluate no further questions ($n = 0$), it returns $\inl q$.
If $\delta$ has to evaluate $\succN n$ questions,
it evaluates $f q \hasvalue a$ and 
recurses on the subtree of $\sigma$ with answer $a$, i.e.\ on $\sigma \at [a]$.
We first verify that $\delta$ composes with interrogations by induction on the interrogation:

\begin{lemma}[][interrogation_plus]
  If
  $\interrogate \sigma {(\lambda q a.\,f q \hasvalue a)} \qs \ans$
  and
  $\delta (\tau \at \ans) f n \hasvalue v$
  then
  $\delta \tau f n \hasvalue v$.
\end{lemma}

Conversely, every evaluation of $\delta$ yields a correct interrogation:

\begin{lemma}[][evalt_to_interrogation]
  If $\delta\;\sigma\;f\;n \hasvalue \inr o$
  then there are $\qs$ and $\ans$ with
  $|\qs| \leq n$ and $\interrogate \sigma {(\lambda q a.\,f q \hasvalue a)} \qs \ans$,
  and $\sigma\;\ans \hasvalue \inr o$.
\end{lemma}
\begin{proof}
  By induction on $n$, using \Cref{lem:noqinterrogation_app}. \qed
\end{proof}

Put together, a computable functional is fully captured by $\delta$ for oracles described by partial functions:

\begin{lemma}[][interrogation_equiv_evalt]
  Given a functional $F$ computed by $\tau$
  we have that 
  \[ F (\lambda q a.\,f q \hasvalue a) i o \leftrightarrow \exists n.\;\delta\;(\tau i)\;f\;n \hasvalue \inr o. \]
\end{lemma}

This is enough preparation to describe the desired computational core of computable functionals:

\begin{theorem}[][Turing_transports_computable]
  If $F$ maps $R\of{Q \to A \to \Prop}$ to a relation $\ofbox{I \to O \to \Prop}$
  and $F$ is computable,
  then there is a partial function $f \of {(Q \pfun A) \to I \pfun O}$ such that if
  $R$ is computed by a partial function $r \of{ Q \pfun A}$, then
  $F R$ is computed by $f r$.
\end{theorem}
\begin{proof}
  Let $F$ be computed by $\tau$.
  We define $f r i$ to search for $n$ such that $\delta\;(\tau i)\;f\;n$ returns $\inr o$, and let it return this $o$.
  The claim then follows straightforwardly by the previous lemma and \Cref{coq:interrogation_output_det}.
  \qed
\end{proof}

\section{Properties of Oracle Semi-Decidability}
\label{sec:prop_sdec}

In the following two sections we establish some standard properties of our synthetic renderings of oracle semi-decidability and Turing reducibility, respectively.
All proofs are concise but precise, given that in the synthetic setting they just amount to the essence of the computational manipulations often described just informally for a concrete model of computation in the analytic approach to computability employed e.g.\ in textbooks.

We first establish the connection to non-relative semi-decidability.

\setCoqFilename{TuringReducibility.SemiDec}
\begin{lemma}[][semi_decidable_to_OracleSemiDecidable]
	If $p$ is semi-decidable, then $\sdec q p$ for any $q$.
\end{lemma}
\begin{proof}
	Let $f \of{X \pfun \mathbbm{1}}$ be a semi-decider for $p$.
	With \Cref{lem:allaboutcomp} (2) the functional mapping $R$ to $\lambda x o.\;f x \hasvalue o$ is computable, and it is easily shown to be a semi-decider for $p$ relative to $q$.
  \qed
\end{proof}

\begin{lemma}[][OracleSemiDecidable_semi_decidable]\label{lem:sdec_dec}
	If $\sdec q p$ and $q$ is decidable, then $p$ is semi-decidable.
\end{lemma}
\begin{proof}
	Let $g$ decide $q$ and let $F$ be a semi-decider of $p$ relative to $q$.
	Let $f$ be the function from \Cref{coq:Turing_transports_computable} that transports computable functions along $F$.
	Now $f (\lambda y.\;\mathsf{ret} (g y))$ is a semi-decider for $p$. \qed
\end{proof}

We next establish closure properties of oracle semi-decidability along reductions.
First, we can replace the oracle by any other oracle it reduces to:

\begin{lemma}[][Turing_transports_sdec]\label{lem:trans1}
	If $\sdec{q}p$ and $q \redT q'$, then also $\sdec {q'} p$.
\end{lemma}
\begin{proof}
	Straightforward using \Cref{coq:computable_comp}. \qed
\end{proof}

Secondly, if we can semi-decide a predicate $p$ relative to $q$, then also simpler predicates should be semi-decidable relative to $q$.
This however requires a stricter notion of reduction, for instance many-one reductions that rule out complementation.
As in~\cite{forster2019synthetic}, we say that $p':X\to \Prop$ many-one reduces to $p:Y\to \Prop$ if there is a function $f:X\to Y$ embedding $p'$ into $p$:\label{intro:redm}
\[ p' \redm p ~:=~ \exists f:X \to Y.\;\forall x.\; p' x \leftrightarrow p (f x) \]

Now the sought after property can be stated as follows:

\begin{lemma}[][red_m_transports_sdec]\label{lem:trans2}
	If $\sdec q {p}$ and $p' \redm p$, then also $\sdec q {p'}$.
\end{lemma}
\begin{proof}
	Straightforward using \Cref{lem:allaboutcomp} (1,4) and \Cref{coq:computable_bind}. \qed
\end{proof}

\section{Properties of Turing Reducibility}
\label{sec:prop_turing}

We continue with similarly standard properties of Turing reducibility.
Again, all proofs are concise but precise.
As a preparation, we first note that Turing reducibility can be characterised without the relational layer.

\setCoqFilename{TuringReducibility.OracleComputability}
\begin{lemma}[][Turing_reducible_without_rel]
  $p \preceq_{\mathsf{T}} q$ if and only if there
  is $\tau$ such that for all $x$ and $b$
  we have \[\charrel p x b \leftrightarrow \exists \qs \ans.\; \interrogate {\tau x} q \qs \ans \land
    \tau\,x\,\ans \hasvalue \out b.\]
\end{lemma}

Now to begin, we show that Turing reducibility is a preorder.

\setCoqFilename{TuringReducibility.OracleComputability}
\begin{theorem}[][Turing_refl]
  Turing reducibility is reflexive and transitive.
\end{theorem}
\begin{proof}
  Reflexivity follows directly by the identity functional being computable via \Cref{lem:allaboutcomp}.
  Transitivity follows with \Cref{coq:computable_comp}. \qed
\end{proof}

In fact, Turing reducibility is an upper semilattice:

\begin{theorem}[][Turing_upper_semi_lattice]
  Let $p \of{ X \to \Prop}$ and $q \of {Y \to \Prop}$.
  Then there is a lowest upper bound $p + q \of{\sumt X Y \to \Prop}$ w.r.t.\ $\redT$:
  Let $(p + q) (\mathsf{inl}\; x) := p x$ and $(p + q) (\mathsf{inr}\; y) := q y$.
  then $p + q$ is the join of $p$ and $q$ w.r.t $\redT$, i.e.\
  $p \redT p + q$, $q \redT p + q$, and for all $r$ if $p \redT r$ and $q \redT r$ then $p + q \redT r$.
\end{theorem}
\begin{proof}
  The first two claims follow by \Cref{lem:allaboutcomp} (1) and \Cref{coq:computable_id}.
  For the third, let $F_1$ reduce $p$ to $r$ and be computed by $\tau_1$ and $F_2$ reduce $q$ to $r$ computed by $\tau_2$.
  Define
  \begin{mathpar}
    F R\,z\,o := \begin{cases}
                       F_1 R \,x \,o & \text{if } z = \mathsf{inl}\; x \\
                       F_2 R \,x \,o & \text{if } z = \mathsf{inr}\; y
                     \end{cases} \and

    \tau z l :=
    \begin{cases}
      \tau_1 x l & \text{if } z = \mathsf{inl}\; x \\
      \tau_2 y l & \text{if } z = \mathsf{inr}\; y
    \end{cases}
  \end{mathpar}
  $\tau$ computes $F$, and $F$ reduces $p + q$ to $r$.
\qed
\end{proof}

We continue by establishing properties analogous to the ones concerning oracle semi-decidability discussed in \Cref{sec:prop_sdec}.
First, analogously to \Cref{coq:semi_decidable_to_OracleSemiDecidable}, the non-relativised notion of decidability implies %
Turing reducibility:

\setCoqFilename{TuringReducibility.SemiDec}
\begin{lemma}[][bisemidecidable_Turing]
	If $p$ and $\overline p$ are semi-decidable, then $p\redT q$ for any $q$.
	In particular, if $p$ is decidable, then $p\redT q$ for any $q$.
\end{lemma}
\begin{proof}
  Let $f$ semi-decide $p$ and $g$ semi-decide $\overline p$.
  Define $F R\,x\,b := \charrel p x b$ and let $\tau x l$ ignore $l$ and find the least $n$ such that either $f x n = \btrue$ or $g x n = \btrue$ and then return $\inr (f x n)$.
  \qed
\end{proof}

Secondly, \Cref{lem:trans1,lem:trans2} correspond to the transitivity of Turing reducibility, the latter relying on the fact that many-one reductions induce Turing reductions:

\setCoqFilename{TuringReducibility.OracleComputability}
\begin{lemma}[][red_m_impl_red_T]
	If $p\redm q$ then $p \redT q$.
\end{lemma}
\begin{proof}
  Let $f$ be the many-one reduction.
  Define $F R\,x\,b := R\,(f x)\,b$.
  \qed
\end{proof}

Thirdly, in connection to \Cref{lem:sdec_dec}, we prove the more involved result that Turing reducibility reflects decidability if and only if Markov's principle holds.
Markov's principle is an axiom in constructive mathematics stating that satisfiability of functions $\ofbox{\nat\to\bool}$ is stable under double negation, i.e.:
\[ \mathsf{MP} := \forall f\of{\nat\to\bool}.\;\neg\neg(\exists n.\;f n = \btrue) \to \exists n.\;f n = \btrue \]

Concretely, $\mathsf{MP}$ will be needed as it corresponds to the termination of non-diverging partial functions:

\begin{lemma}\label{coq:MP_to_MP_partial}
  $\mathsf{MP}$ if and only if $\forall X Y. \forall f \of{ X \pto Y}.\forall x.\;\neg\neg(\exists y.\;f x \hasvalue Y) \to \exists y.\;f x \hasvalue Y$.
\end{lemma}

Another ingredient is that total partial function $X\pto Y$ induce functions $X\to Y$, as stated here for the specific case of deciders $X\to \bool$:

\begin{lemma}[][partial_decidable]
  Let $f\of{X\pfun\bool}$ and
  $p\of{X\to\Prop}$.
  If $\forall x.\;p x \leftrightarrow f x \hasvalue \btrue$
  and $\forall x.\exists b.\;f x \hasvalue b$, then
  $p$ is decidable, i.e.\ there is a function $g\of{X\to\bool}$ such that $\forall x.\;p x \leftrightarrow g x = \btrue$.
\end{lemma}

Now assuming $p \redT q$ for $q$ decidable, we can derive a non-diverging partial decider for $p$, which is turned into a total partial decider with \Cref{coq:MP_to_MP_partial} and then into an actual decider with \Cref{coq:partial_decidable}:

\begin{theorem}[][transport_decidable]
  Given $\mathsf{MP}$,
  if $q$ is decidable and $p \redT q$, then $p$ is decidable.
\end{theorem}
\begin{proof}
  Let $F$ be the reduction relation and let $f$ transport computability along it as in \Cref{coq:Turing_transports_computable}.
  Let $g$ decide $q$.
  It is straightforward that $\forall x b.\;\charrel p x b \leftrightarrow f (\lambda y.\ret (g y)) x \hasvalue b$ (*).
  It suffices to prove that $\forall x.\exists b.\;f (\lambda y.\ret (g y)) x\hasvalue b$ to obtain the claim from \Cref{coq:partial_decidable}.

  Using \Cref{coq:MP_to_MP_partial} and $\mathsf{MP}$, given $x$ it suffices to prove
  $\neg\neg\exists b.\;f (\lambda y.\ret (g y)) x\hasvalue b$.
  Because the goal is negative and we can prove $\neg\neg(p x \lor \neg p x)$,
  we are allowed to do a case analysis on $p x$.
  In both cases we can prove termination using (*).
  \qed
\end{proof}

As hinted above, the previous theorem could be stated without $\MP$ by using a notion of decidability via a non-diverging partial decider $f\of{X\pfun\bool}$, i.e.\ with $\forall x.\neg\neg\exists b.\;f x \hasvalue b$.
However, in the stated form, it is in fact equivalent to $\MP$:

\begin{lemma}[][decidable_Turing_MP]
	\label{lem_MP_reverse}
  If $p$ is decidable if there is decidable $q$ with $p \redT q$, then $\MP$ holds.
\end{lemma}
\begin{proof}
  By~\cite[Theorem 2.20]{forster2019synthetic} it suffices to prove that whenever $p\of{\nat\to\Prop}$ and $\overline p$ are semi-decidable, then also $p$ is decidable, which follows by \Cref{coq:bisemidecidable_Turing} and the assumption for some choice of a decidable predicate $q$.
  \qed
\end{proof}

Lastly, we prove that using classical logic, predicates are Turing-equivalent to their complement, providing evidence for the inherent classicality:

\begin{lemma}[][Turing_red_compl]
  For double-negation stable $p$, $p \redT \overline p$ and $\overline p \redT p$.
\end{lemma}
\begin{proof}
  Assume $\forall x.\;\neg\neg p x \to p x$.
  For both reductions, take $F R\,x\,b := R\,x\,(\neg_{\bool} b)$,
  which is computable by \Cref{coq:computable_bind}, \Cref{coq:computable_id}, and \Cref{lem:allaboutcomp} (1,3).
  \qed
\end{proof}

\begin{lemma}[][rev]
	\label{lem_LEM_reverse}
  Let $X$ be some type with $x_0 \of X$.
  If $p \redT \overline p$ for all $p\of{X\to\Prop}$, then $\mathsf{MP}$ implies the law of excluded middle $(\mathsf{LEM}:=\forall P:\Prop.\,P\lor\neg P )$.
\end{lemma}
\begin{proof}
  Assume $\mathsf{MP}$, $X$ with $x_0:X$, and that $p \redT \overline p$ for all $p:X\to\Prop$.
  It suffices to prove that for every proposition $P$ we have $\neg\neg P \to P$.
  So assume $\neg\neg P$.

  By $\mathsf{MP}$ and \Cref{coq:transport_decidable},
  we have that whenever $\lambda x.\,\neg P$ is decidable, then so is $\lambda x.\,P$.
  Now since $\neg\neg P$ holds, $\lambda x.\,\bfalse$ decides $\lambda x.\,\neg P$.
  Thus we have a decider $f$ for $\lambda x.\,P$.
  A case analysis on $f x_0$
  yields either $P$ and we are done -- or $\neg P$, which is ruled out by $\neg\neg P$.
  \qed
\end{proof}

The last lemma ensures that some amount of classical logic is necessary to prove that Turing reducibility is closed under complements,
since it is well-known that $\mathsf{MP}$ does not imply $\mathsf{LEM}$.

\section{Turing Reducibility and Truth-Table Reducibility}
\label{sec:tt}

As a further expectable property, we establish the well-known connection of Turing reducibility to truth-table reducibility, namely that every truth-table reduction induces a Turing reduction while the converse does not hold.
Note that the proofs in this section have a classical flavour where explicitly mentioned.

We use the synthetic definition of truth-table reducibility from Forster and Jahn~\cite{csl23}.
We model truth-tables as lists $\List \bool$, but just work with a boolean evaluation predicate $l \vDash T$ and refer to the Coq code for its definition.
\label{intro:redtt}
\newcommand\redtt{\hyperref[def:redtt]{\preceq_{\texttt{tt}}}}
\begin{align*}
  p \redtt q := \exists f \of {X \to \List Y \!\!\times\!\List \bool}.\forall x \of X.\forall l \of{\List\bool}.\,\Forall_2\,\charrel q\,(\pi_1 (f x))\,l \to (p x \leftrightarrow l \vDash \pi_2 (f x) )
\end{align*}
where $\Forall_2$ lifts binary predicates to lists pointwise by conjunction.

We first show that truth-table reducibility implies Turing reducibility.

\begin{theorem}[][truthtable_Turing]
  If $q$ is classical (i.e.\ $\forall y.\;q y \lor \neg q y$),
  then
  $p \redtt q$ implies $p \redT q$.
\end{theorem}
\begin{proof}
  Let $f$ be the truth-table reduction.
  Define $F$ to map $R\of{Y \to \bool \to \Prop}$ to
  \begin{align*}
    \lambda x b.\; \exists l \of{\List\bool}.\; \Forall_2\;R\;(\pi_1 (f x))\;l \land 
                                l \vDash \pi_2 (f x))
  \end{align*}
  which can be computed by the tree
  \begin{mathpar}
    \tau x l :=
    \begin{cases}
      \ret (\inl a) & \text{if } \pi_1 (f x) \text{ at position } |l| \text{ is } a \\
      \ret (\inr (l \vDash \pi_2 (f x))) & \text{otherwise.}
    \end{cases}
  \end{mathpar}

  The direction from right to left is straightforward.
  For the direction from left to right, it suffices to prove the existence of $l$ with $\Forall_2\;\charrel q\;\pi_1(f x)\;l$,
  following by induction on $\pi_1 (f x)$, using the assumption that $q$ is classical to construct $l$.  
  \qed
\end{proof}

We now prove that the inclusion of truth-table reducibility in Turing reducibility is strict.
Forster and Jahn~\cite{csl23} introduce a hypersimple predicate $H_I\of{\nat\to\Prop}$ as the deficiency predicate of a strongly enumerable predicate $I\of{\nat\to\Prop}$~\cite{dekker1954hypersimple}:
Given an injective, strong enumerator $E_I%
$ of $I$ ($\forall x.\;I x\! \leftrightarrow\! \exists n. E_I n = x$), they set
\[H_I x ~:=~ \exists x_0 > x.\; E_I x_0 < E_I x.\]

They prove that $I$ does not truth-table reduce to $H_I$ assuming axioms for synthetic computability, and in particular that the halting problem fulfills the preconditions for $I$.
Thus, to separate truth-table from Turing reducibility, it suffices to give a Turing reduction $I \redT H_I$ (without having to assume axioms for synthetic computability).

Algorithmically, one can decide $I z$ given a partial function $f\of{\nat\pfun\bool}$ deciding $H_I$ as follows:
We search for $x$ such that $f x \hasvalue \bfalse$ and $E_I x > z$, i.e.\ $\neg H_I x$.
Such an $x$ does (not not) exists because the complement of $H_I$ is non-finite.
Then $I z$ holds if and only if $z \in [E_I 0, \dots E_I (x + 1)]$.

Formally, we first establish the classical existence of such $x$ in the more general situation of arbitrary non-finite predicates and injections.

\begin{lemma}[][non_finite_to_least]
  If $p\of{X\to\Prop}$ is non-finite and $f\of{X \to \nat}$ is injective, then for $z\of\nat$%
  \[\neg\neg\exists x.\;p x \land f x \geq z \land \forall y.\; p y \to f y \geq z \to f x \leq f y. \]
\end{lemma}

Next, we verify the resulting characterisation of $I$ via list membership.

\begin{lemma}[][I_iff]
  If $\neg H_I x$ and $E_I x > z$, then $I z \leftrightarrow [E_I 0, \dots, E_I (x + 1)]$.
\end{lemma}

Put together, we can describe the desired Turing reduction.

\begin{theorem}[][red]
  Assuming $\mathsf{LEM}$,
  if $I$ is strongly enumerable,
  then
  $I \redT H_I$.
\end{theorem}
\begin{proof}
  We define $F$ to map relations $R$ to the relation
  \begin{align*}
    \lambda z b .\; \exists x.\, R\;x\; \bfalse ~\land~ E_I x > z &~ \land ~                                    (b = \btrue \leftrightarrow z \in [E_I 0, \dots, E_I (x + 1)])  \\
                                  &~\land~  (\forall x' < x.\, (R\; x'\; \btrue \lor (R\; x'\;\bfalse \land E_I x' \leq z))) 
  \end{align*}
  which is straightforward to show computable.

  Regarding $F (\charrel{H_I}) z b \leftrightarrow \charrel I z b$,
  the direction from left to right is immediate from \Cref{coq:I_iff}.
  For the direction from right to left, assume $\charrel I z b$.
  Let $x$ be obtained for $H_I$ and $E_I$ from \Cref{coq:non_finite_to_least}.
  Then $x$ fulfils the claim by \Cref{coq:I_iff}.
  \qed
\end{proof}

Since in this paper we do not assume axioms for synthetic computability that imply $I \not\redtt H_I$, we keep the conclusion that truth-table reducibility is strictly stronger than Turing reducibility implicit.

\section{Post's Theorem ($\mathsf{PT}$)}
\label{sec:PT}

There are various results (rightly) called ``Post's theorem'' in the literature.
Here, we are concerned with the result that if both a predicate and its complement are semi-decidable,
the predicate is decidable.
This theorem was proved by Post in 1944~\cite{post1944recursively}, and is not to be confused with Post's theorem relating the arithmetical hierarchy and Turing jumps from 1948~\cite{post1948degrees}. %
We thus simply refer to the result we consider as $\mathsf{PT}_0$,
and use $\mathsf{PT}$ for its relativised version.

It is well-known that $\mathsf{PT}_0$ is equivalent to Markov's principle~\cite{troelstra1988constructivism,BauerSyntCT,forster2019synthetic}.
We here prove that the relativised version $\mathsf{PT}$ is fully constructive,
and that in fact the equivalence proof of $\mathsf{MP}$ and $\mathsf{PT}_0$ can be given using $\mathsf{PT}$ and the already proven equivalence between $\mathsf{MP}$
and the statement that Turing reducibility transports decidability backwards given in \Cref{sec:prop_turing}.

As an auxiliary notion, we introduce an equivalent but a priori more expressive form of interrogations which maintains an internal state of the computation and can ``stall'', i.e.\
trees do not have to either ask a question or produce an output, but can alternatively choose to just update the state.
Such trees are of type $S \to \List A \pto \sumt {(S \times \option Q)} O$,
where $\option Q$ is the inductive option type with elements $\None$ and $\Some q$ for $q \of Q$\label{intro:option}.

A stalling tree is a partial function $\sigma \of{S \to \List A \pto \sumt {(S \times \option Q)} O}$.
We define a stalling interrogation predicate $\interrogate \sigma R \qs \ans \mathbin{;} s \succ s'$ inductively by:
\begin{mathpar}
	\infer{~}{\interrogate \sigma R {[]} {[]}  \mathbin{;} s \succ s}
	
	\infer{\interrogate \sigma R \qs \ans  \mathbin{;} s \succ s'' \and \sigma\ \mathbin{;}s''\ \mathbin{;}\ans \hasvalue \inl (s',\None)}
	{\interrogate \sigma R {\qs} {\ans}  \mathbin{;} s \succ s'}
	
	\infer{\interrogate \sigma R \qs \ans  \mathbin{;} s \succ s'' \and \sigma\ \mathbin{;}s''\ \mathbin{;}\ans \hasvalue \inl (s',\Some q) \and R q a}
	{\interrogate \sigma R {\qs \app [q]} {\ans \app [a]}  \mathbin{;} s \succ s'}
\end{mathpar}
The first and third rule are not significantly different from before, apart from also threading a state $s$.
The second rule allows the tree to stall by only updating the state to $s'$, but without asking an actual question.
Intuitively, we can turn a stalling tree $\tau$ into a non-stalling one $\tau'$ by having $\tau'$
compute on input $\ans$ first all results of $\tau$ on all prefixes of $\ans$, starting from a call $\tau\,i\,s_0\;\ans$ for a given initial state $s_0$.
We give this construction in full detail in \Cref{sec:forms}.

A functional $F$
mapping $R \of{Q \to A \to \Prop}$
to a relation of type
$ \ofbox{I \to O \to \Prop}$ is computable via stalling interrogations if
there are a type $S$, an element $s_0 \of S$, and
a function $\tau \of{I \to S \to \List{A} \pto \sumt {(S \times \option Q)} O}$ such that
\[ \forall R\,i\,o.\;F R\,i\,o \leftrightarrow \exists \qs \; \ans \; s. ~ \interrogate {\tau i} R \qs \ans \mathbin{;} s_0 \succ s ~\land~ \tau\,i\,s\;\ans \hasvalue \inr o.  \]

We prove that the two definitions of computability are equivalent in \Cref{sec:forms}
and immediately move on to the proof of $\mathsf{PT}$.

\setCoqFilename{TuringReducibility.SemiDec}
\begin{theorem}[$\mathsf{PT}$][PT]
  If $\sdec q p$ and $\sdec q {\overline p}$, then
  $p \redT q$.
\end{theorem}
\begin{proof}
  Let $p\of{X \to \Prop}$ and $q\of{Y\to\Prop}$ as well as $F_1$ and $F_2$ be the functionals representing the semi-deciders,
  computed respectively by $\tau_1$ and $\tau_2$.
  The intuition is, on input $x$ and $\ans$, to execute $\tau_1\;x$ and $\tau_2\;x$ in parallel and ensure that both their questions are asked.
  The interrogation can finish with $\btrue$ if $\tau_1\;x$ outputs a value,
  and with $\bfalse$ if $\tau_2\;x$ does.

  There are two challenges in making this intuition formal as an oracle computation:
  Only answers from $\ans$ that $\tau_{1}$ and $\tau_2$ asked for have to be actually passed to it, respectively,
  and both $\tau_1$ and $\tau_2$ need to be allowed to ask all of their questions and eventually produce an output fairly, even though only one of them ever will.
  
  	\newcommand{\getans}[1]{\mathsf{getas}_{#1}}
  	\newcommand{\getqs}[1]{\mathsf{getqs}_{#1}}

  Using \Cref{coq:Turing_reducible_without_rel}, we define the Turing reduction without providing the relational layer and instead directly construct a tree $\tau$ based on stalling interrogations
  with state type $S := \option Y \times \nat \times \List{(\bool \times Y)}$.
  The first argument is used to remember a question that needs to be asked next, arising from cases where both $\tau_1$ and $\tau_2$ want to ask a question.
  The second argument is a step-index $n$ used to evaluate both $\tau_1$ and $\tau_2$ for $n$ steps.
  The third argument records which question was asked by $\tau_1$ and which by $\tau_2$.
  To then construct $\tau$ compactly, we define helper functions $\getans{1,2} \of{\List{(\bool \times Y)} \to \List \bool \to \List Y}$ which choose answers from the second list according to the respective boolean in the first list.
  
  We then define
  \begin{mathpar}
    \tau (\Some q, n, t) \ans := \ret (\inl (\None, n, t \app [(\bfalse, q)], \Some q)) \\
    
    \tau (\None, n, t) \ans :=\scriptsize
    \begin{cases}
      \ret(\inr \btrue) & \text{if } x_1 = \Some(\inr o)  \\
      \ret(\inr \bfalse) & \text{if } x_2 = \Some (\inr o) \\
      \ret(\inl (\Some q', \succN n, t \app [(\btrue, q)], \Some q))
                        & \text{if } x_1 = \Some (\inl q) \\
      & \text{ and } x_2 = \Some (\inl q')  \\
      \ret(\inl (\None, \succN n, t \app [(\btrue, q)], \Some q))
                        & \text{if } x_1 = \Some (\inl q)  \\
      \ret(\inl (\None, \succN n, t \app [(\bfalse, q)], \Some q))
                        & \text{if } x_2 = \Some (\inl q)  \\
      \ret(\inl (\None, \succN n, t, \None)) & \text{otherwise} 
    \end{cases}
  \end{mathpar}
  where $x_1 = \rho^n\,(\tau_1\, x\, (\getans{1}\, t\, \ans))$ and $x_2 = \rho^n\,(\tau_2\, x\, (\getans{2}\, t\, \ans))$,
  with $\rho$ being a step-indexed evaluation function for partial values.

  This means that whenever $\tau_1$ returns an output, then $\btrue$ is returned and
  whenever $\tau_2$ returns an output, then $\bfalse$ is returned while no question is ever missed
  and the interrogation stalls if $n$ does not suffice to evaluate either $\tau_1$ or $\tau_2$.
  The invariants to prove that this indeed yields the wanted Turing reduction are technical but pose no major hurdles, we refer to the Coq code for details.\qed

\end{proof}

\enlargethispage{1.5\baselineskip}%
\begin{corollary}	\label{corollaryPT}
  The following are equivalent:
  \begin{enumerate}
  \item $\mathsf{MP}$
  \item Termination of partial functions is double negation stable.
  \item Turing reducibility transports decidability backwards.
  \item $\mathsf{PT}_0$
  \end{enumerate}
\end{corollary}
\begin{proof}
  Implications $(1) \to (2)$ and $(4) \to (1)$ are well-known.
  We have already proved implication $(2) \to (3)$.
  It suffices to prove $(3) \to (4)$, which is almost direct using~$\mathsf{PT}$:
  Assume that for all $X$, $Y$, $p\of{X\to\Prop}$, and $q\of{Y\to\Prop}$
  we have that if $q$ is decidable and $p \redT q$, then $p$ is decidable.
  Let furthermore $p$ and its complement be semi-decidable.
  We prove that $p$ is decidable.
  Clearly, it suffices to prove that $p \redT q$ for a decidable predicate $q$ (e.g.\ $\lambda n \of \nat.\top$).
  Using $\mathsf{PT}$, it suffices to prove $p$ and its complement semi-decidable in $q$,
  which in turn follows from the assumption that they are semi-decidable and \Cref{coq:semi_decidable_to_OracleSemiDecidable}.
  \qed
\end{proof}

\section{Discussion}
\label{sec:conclusion}

\textit{Mechanisation in Coq}\quad
\label{sec:coq}
The Coq mechanisation %
accompanying this paper closely follows the structure of the hyperlinked mathematical presentation and spans roughly 2500 lines of code for the novel results, building on a library of basic synthetic computability theory.
It showcases the feasibility of mechanising ongoing research with reasonable effort and illustrates the interpretation of synthetic oracle computations as a natural notion available in dependently-typed programming languages.
In fact, using Coq helped us a lot with finding the proofs concerning constructive reverse mathematics (\Cref{lem_MP_reverse,lem_LEM_reverse} and \Cref{corollaryPT}) in the first place, where subtleties like double negations need to be tracked over small changes in the definitions.

On top of the usual proof engineering, we used three notable mechanisation techniques.
First, we generalise over all possible implementations of partial functions, so our code is guaranteed to just rely on the abstract interface described in \Cref{sec:glossary}.
Secondly, we devised a custom tactic \texttt{psimpl} that simplifies goals involving partial functions by strategically rewriting with the specifications of the respective operations.
Thirdly, to establish computability of composed functionals, instead of constructing a complicated tree at once, we postpone the construction with the use of existential variables and apply abstract lemmas such as the ones described in \Cref{sec:closure} to obtain the trees step by step.

\textit{Related Work}\quad
Synthetic computability was introduced by Richman~\cite{richman1983church} and popularised by Richman, Bridges, and Bauer~\cite{bridges1987varieties,BauerSyntCT,bauer2017fixed,bauer2020Wisc}.
In synthetic computability, one assumes axioms such as
$\mathsf{CT}$ (``Church's thesis''~\cite{kreisel1965mathematical,troelstra1988constructivism}), postulating that \textit{all} functions are $\mu$-recursive.
$\mathsf{CT}$ is proved consistent for univalent type theory by Swan and Uemura~\cite{swan2019church}.
Since univalent type theory proves unique choice, using it as the basis for computability theory renders $\mathsf{CT}$ inconsistent with already the weak principle of omniscience~\cite{forster2020churchs}, and consequently with the law of excluded middle, precluding interesting results in constructive reverse mathematics.

Forster~\cite{forster2022parametric} identifies that working in CIC allows assuming $\mathsf{CT}$ and its consequences even under the presence of the law of excluded middle.
This approach has been used to develop the theory of many-one and truth-table reducibility~\cite{forster_et_al:LIPIcs.CSL.2023.21},
to give a proof of the Myhill isomorphism theorem~\cite{forster:hal-03891390} and a more general treatment of computational back-and-forth arguments~\cite{kirst2022computational}, to show that random numbers defined using Kolmogorov complexity form a simple set~\cite{forster_et_al:LIPIcs.ITP.2022.12},
to analyse Tennenbaum's theorem regarding its constructive content~\cite{hermes_et_al:LIPIcs.FSCD.2022.9},
to give computational proofs of Gödel's first incompleteness theorem~\cite{kirst2023synthetic,kirst2023godel},
and to develop an extensive Coq library of undecidability proofs~\cite{forster2020coq}.

The first synthetic definition of oracle computability is due to Bauer~\cite{bauer2020Wisc}, based on continuous functionals in the effective topos.
Forster has introduced a classically equivalent definition in his PhD thesis~\cite{forster2022thesis} based on joint work with Kirst.
Forster and Kirst have adapted this definition into one constructively equivalent to Bauer's definition~\cite{types22abstract}.
All these previous definitions however have in common that it is unclear how to derive an enumeration of all oracle computable functionals from $\mathsf{CT}$ as used in \cite{types22kleenepost}, because they do no reduce higher-order functionals to first-order functions.
Recently, Swan has suggested a definition of oracle computability based on modalities in univalent type theory~\cite{swanoracle}.

\textit{Future Work}\quad
With the present paper, we lay the foundation for several future investigations concerning synthetic oracle computability in the context of axioms like $\mathsf{CT}$, both by improving on related projects and by tackling new challenges.
First, a rather simple test would be the Kleene-Post theorem~\cite{Kleene1954}, establishing incomparable Turing degrees as already approximated in~\cite{types22kleenepost}, assuming an enumeration of all oracle computations of their setting. %
Similarly, we plan to establish Post's theorem~\cite{post1948degrees}, connecting the arithmetical hierarchy with Turing degrees. %
An interesting challenge would be a synthetic proof of the Friedberg-Muchnik theorem \cite{Friedberg1957,muchnik1963strong}, solving Post's problem~\cite{post1944recursively} concerning the existence of undecidable Turing degrees strictly below the halting problem.

\subsubsection{Acknowledgements} We want to thank Felix Jahn, Gert Smolka, Dominique Larchey-Wendling, and the participants of the TYPES\ '22 conference for many fruitful discussions about Turing reducibility,
as well as Martin Baillon, Yann Leray, Assia Mahboubi, Pierre-Marie Pédrot, and Matthieu Piquerez for discussions about notions of continuity.
The central inspiration to start working on Turing reducibility in type theory is due to Andrej Bauer's talk at the Wisconsin logic seminar in February 2021.

\bibliographystyle{splncs04}
\bibliography{biblio.bib}

\begin{thebibliography}{10}
\providecommand{\url}[1]{\texttt{#1}}
\providecommand{\urlprefix}{URL }
\providecommand{\doi}[1]{https://doi.org/#1}

\bibitem{BauerSyntCT}
Bauer, A.: First steps in synthetic computability theory. Electronic Notes in
  Theoretical Computer Science  \textbf{155},  5--31 (2006).
  \doi{10.1016/j.entcs.2005.11.049}

\bibitem{bauer2017fixed}
Bauer, A.: On fixed-point theorems in synthetic computability. Tbilisi
  Mathematical Journal  \textbf{10}(3),  167--181 (2017).
  \doi{10.1515/tmj-2017-0107}

\bibitem{bauer2020Wisc}
Bauer, A.: Synthetic mathematics with an excursion into computability theory
  (slide set). University of Wisconsin Logic seminar  (2020),
  \url{http://math.andrej.com/asset/data/madison-synthetic-computability-talk.pdf}

\bibitem{bridges1987varieties}
Bridges, D., Richman, F.: Varieties of constructive mathematics, vol.~97.
  Cambridge University Press (1987). \doi{10.1017/CBO9780511565663}

\bibitem{coquand1986calculus}
Coquand, T., Huet, G.P.: The calculus of constructions. Information and
  Computation  \textbf{76}(2/3),  95--120 (1988).
  \doi{10.1016/0890-5401(88)90005-3},
  \url{https://doi.org/10.1016/0890-5401(88)90005-3}

\bibitem{Davis1958-DAVCU}
Davis, M.D.: Computability and Unsolvability. McGraw-Hill Series in Information
  Processing and Computers, McGraw-Hill (1958)

\bibitem{dekker1954hypersimple}
Dekker, J.C.E.: A theorem on hypersimple sets. Proceedings of the American
  Mathematical Society  \textbf{5},  791--796 (1954).
  \doi{10.1090/S0002-9939-1954-0063995-6}

\bibitem{escardodialogue}
Escardo, M.: Continuity of {G}ödel's system {T} definable functionals via
  effectful forcing. Electronic Notes in Theoretical Computer Science
  \textbf{298},  119–141 (11 2013). \doi{10.1016/j.entcs.2013.09.010}

\bibitem{forster2020churchs}
Forster, Y.: {Church’s Thesis and Related Axioms in Coq’s Type Theory}. In:
  Baier, C., Goubault-Larrecq, J. (eds.) 29th EACSL Annual Conference on
  Computer Science Logic (CSL 2021). Leibniz International Proceedings in
  Informatics (LIPIcs), vol.~183, pp. 21:1--21:19. Schloss
  Dagstuhl--Leibniz-Zentrum f{\"u}r Informatik, Dagstuhl, Germany (2021).
  \doi{10.4230/LIPIcs.CSL.2021.21},
  \url{https://drops.dagstuhl.de/opus/volltexte/2021/13455}

\bibitem{forster2022thesis}
Forster, Y.: Computability in Constructive Type Theory. Ph.D. thesis, Saarland
  University (2021). \doi{10.22028/D291-35758}

\bibitem{forster2022parametric}
Forster, Y.: Parametric {Church’s Thesis}: Synthetic computability without
  choice. In: International Symposium on Logical Foundations of Computer
  Science. pp. 70--89. Springer (2022). \doi{10.1007/978-3-030-93100-1\_6}

\bibitem{csl23}
Forster, Y., Jahn, F.: {Constructive and Synthetic Reducibility Degrees: Post's
  Problem for Many-one and Truth-table Reducibility in Coq}. In: Klin, B.,
  Pimentel, E. (eds.) 31st EACSL Annual Conference on Computer Science Logic
  (CSL 2023). Leibniz International Proceedings in Informatics (LIPIcs),
  vol.~252, pp. 16:1--16:21. Schloss Dagstuhl--Leibniz-Zentrum fuer Informatik,
  Dagstuhl, Germany (2023). \doi{10.4230/LIPIcs.CSL.2023.16}

\bibitem{forster_et_al:LIPIcs.CSL.2023.21}
Forster, Y., Jahn, F.: {Constructive and Synthetic Reducibility Degrees:
  Post’s Problem for Many-One and Truth-Table Reducibility in Coq}. In: Klin,
  B., Pimentel, E. (eds.) 31st EACSL Annual Conference on Computer Science
  Logic (CSL 2023). Leibniz International Proceedings in Informatics (LIPIcs),
  vol.~252, pp. 21:1--21:21. Schloss Dagstuhl -- Leibniz-Zentrum f{\"u}r
  Informatik, Dagstuhl, Germany (2023). \doi{10.4230/LIPIcs.CSL.2023.21},
  \url{https://drops.dagstuhl.de/opus/volltexte/2023/17482}

\bibitem{forster:hal-03891390}
Forster, Y., Jahn, F., Smolka, G.: {A Computational Cantor-Bernstein and
  Myhill's Isomorphism Theorem in Constructive Type Theory}. In: {CPP 2023 -
  12th ACM SIGPLAN International Conference on Certified Programs and Proofs}.
  pp.~1--8. {ACM}, Boston, United States (Jan 2023).
  \doi{10.1145/3573105.3575690}, \url{https://inria.hal.science/hal-03891390}

\bibitem{types22abstract}
Forster, Y., Kirst, D.: Synthetic {T}uring reducibility in constructive type
  theory. 28th International Conference on Types for Proofs and Programs (TYPES
  2022) (2022),
  \url{https://types22.inria.fr/files/2022/06/TYPES_2022_paper_64.pdf}

\bibitem{forster2019synthetic}
Forster, Y., Kirst, D., Smolka, G.: {On synthetic undecidability in Coq, with
  an application to the Entscheidungsproblem}. In: Proceedings of the 8th {ACM}
  {SIGPLAN} International Conference on Certified Programs and Proofs - {CPP}
  2019. {ACM} Press (2019). \doi{10.1145/3293880.3294091},
  \url{https://doi.org/10.1145/3293880.3294091}

\bibitem{forster_et_al:LIPIcs.ITP.2022.12}
Forster, Y., Kunze, F., Lauermann, N.: {Synthetic Kolmogorov Complexity in
  Coq}. In: Andronick, J., de~Moura, L. (eds.) 13th International Conference on
  Interactive Theorem Proving (ITP 2022). Leibniz International Proceedings in
  Informatics (LIPIcs), vol.~237, pp. 12:1--12:19. Schloss Dagstuhl --
  Leibniz-Zentrum f{\"u}r Informatik, Dagstuhl, Germany (2022).
  \doi{10.4230/LIPIcs.ITP.2022.12},
  \url{https://drops.dagstuhl.de/opus/volltexte/2022/16721}

\bibitem{forster2020coq}
Forster, Y., Larchey-Wendling, D., Dudenhefner, A., Heiter, E., Kirst, D.,
  Kunze, F., Smolka, G., Spies, S., Wehr, D., Wuttke, M.: A {Coq} library of
  undecidable problems. In: The Sixth International Workshop on Coq for
  Programming Languages (CoqPL 2020). (2020),
  \url{https://github.com/uds-psl/coq-library-undecidability}

\bibitem{Friedberg1957}
Friedberg, R.M.: Two recursively enumerable sets of incomparable degrees of
  unsovlability (solution of post's problem, 1944. Proceedings of the National
  Academy of Sciences  \textbf{43}(2),  236--238 (Feb 1957).
  \doi{10.1073/pnas.43.2.236}, \url{https://doi.org/10.1073/pnas.43.2.236}

\bibitem{hermes_et_al:LIPIcs.FSCD.2022.9}
Hermes, M., Kirst, D.: {An Analysis of Tennenbaum’s Theorem in Constructive
  Type Theory}. In: Felty, A.P. (ed.) 7th International Conference on Formal
  Structures for Computation and Deduction (FSCD 2022). Leibniz International
  Proceedings in Informatics (LIPIcs), vol.~228, pp. 9:1--9:19. Schloss
  Dagstuhl -- Leibniz-Zentrum f{\"u}r Informatik, Dagstuhl, Germany (2022)

\bibitem{kirst2022computational}
Kirst, D.: Computational back-and-forth arguments in constructive type theory.
  In: 13th International Conference on Interactive Theorem Proving (ITP 2022).
  Schloss Dagstuhl-Leibniz-Zentrum f{\"u}r Informatik (2022)

\bibitem{types22kleenepost}
Kirst, D., Forster, Y., Mück, N.: {Synthetic Versions of the Kleene-Post and
  Post’s Theorem}. 28th International Conference on Types for Proofs and
  Programs (TYPES 2022) (2022),
  \url{https://types22.inria.fr/files/2022/06/TYPES_2022_paper_65.pdf}

\bibitem{kirst2023synthetic}
Kirst, D., Hermes, M.: Synthetic undecidability and incompleteness of
  first-order axiom systems in {C}oq: Extended version. Journal of Automated
  Reasoning  \textbf{67}(1), ~13 (2023)

\bibitem{kirst2023godel}
Kirst, D., Peters, B.: G{\"o}del’s theorem without tears - essential
  incompleteness in synthetic computability. In: 31st EACSL Annual Conference
  on Computer Science Logic (CSL 2023). Schloss Dagstuhl-Leibniz-Zentrum
  f{\"u}r Informatik (2023)

\bibitem{Kleene_1959}
Kleene, S.C.: Recursive functionals and quantifiers of finite types i.
  Transactions of the American Mathematical Society  \textbf{91}(1), ~1 (Apr
  1959). \doi{10.2307/1993145},
  \url{https://www.jstor.org/stable/1993145?origin=crossref}

\bibitem{kleene1952introduction}
Kleene, S.C.: Introduction to metamathematics, vol.~483. van Nostrand New York
  (1952)

\bibitem{Kleene1954}
Kleene, S.C., Post, E.L.: The upper semi-lattice of degrees of recursive
  unsolvability. The Annals of Mathematics  \textbf{59}(3), ~379 (May 1954).
  \doi{10.2307/1969708}, \url{https://doi.org/10.2307/1969708}

\bibitem{kreisel1965mathematical}
Kreisel, G.: Mathematical logic. Lectures in modern mathematics  \textbf{3},
  95--195 (1965). \doi{10.2307/2315573}

\bibitem{muchnik1963strong}
Muchnik, A.A.: On strong and weak reducibility of algorithmic problems.
  Sibirskii Matematicheskii Zhurnal  \textbf{4}(6),  1328--1341 (1963)

\bibitem{odifreddi1992classical}
Odifreddi, P.: Classical recursion theory: The theory of functions and sets of
  natural numbers. Elsevier (1992)

\bibitem{van_Oosten_1999}
van Oosten, J.: A combinatory algebra for sequential functionals of finite
  type. In: Models and Computability, pp. 389--406. Cambridge University Press
  (jun 1999). \doi{10.1017/cbo9780511565670.019},
  \url{https://doi.org/10.1017%2Fcbo9780511565670.019}

\bibitem{vanOosten2011-VANPCA-2}
van Oosten, J.: Partial combinatory algebras of functions. Notre Dame Journal
  of Formal Logic  \textbf{52}(4),  431--448 (2011).
  \doi{10.1215/00294527-1499381}

\bibitem{paulin1993inductive}
Paulin-Mohring, C.: Inductive definitions in the system {Coq} rules and
  properties. In: International Conference on Typed Lambda Calculi and
  Applications. pp. 328--345. Springer (1993). \doi{10.1007/BFb0037116}

\bibitem{paulin2015introduction}
Paulin-Mohring, C.: {Introduction to the Calculus of Inductive Constructions}
  (Jan 2015), \url{https://hal.inria.fr/hal-01094195}

\bibitem{post1944recursively}
Post, E.L.: Recursively enumerable sets of positive integers and their decision
  problems. bulletin of the American Mathematical Society  \textbf{50}(5),
  284--316 (1944). \doi{10.1090/S0002-9904-1944-08111-1}

\bibitem{post1948degrees}
Post, E.L.: Degrees of recursive unsolvability - preliminary report. In:
  Bulletin of the American Mathematical Society. vol.~54:7, pp. 641--642.
  American Mathematical Society ({AMS}) (1948)

\bibitem{richman1983church}
Richman, F.: Church's thesis without tears. The Journal of symbolic logic
  \textbf{48}(3),  797--803 (1983). \doi{10.2307/2273473}

\bibitem{swan2019church}
Swan, A., Uemura, T.: On {C}hurch's thesis in cubical assemblies. arXiv
  preprint arXiv:1905.03014  (2019), \url{https://arxiv.org/abs/1905.03014}

\bibitem{swanoracle}
Swan, A.W.: Oracle modalities. Second International Conference on Homotopy Type
  Theory (HoTT 2023)  (2023),
  \url{https://hott.github.io/HoTT-2023/abstracts/HoTT-2023_abstract_35.pdf}

\bibitem{Coq}
{The Coq Development Team}: The coq proof assistant version 8.13.2 (Jan 2021).
  \doi{10.5281/zenodo.4501022}, \url{https://doi.org/10.5281/zenodo.4501022}

\bibitem{troelstra1988constructivism}
Troelstra, A.S., van Dalen, D.: Constructivism in mathematics. vol. i. Studies
  in Logic and the Foundations of Mathematics  \textbf{26} (1988)

\bibitem{turing1939systems}
Turing, A.M.: Systems of logic based on ordinals. Proceedings of the London
  mathematical society  \textbf{2}(1),  161--228 (1939).
  \doi{10.1112/plms/s2-45.1.161}

\end{thebibliography}

\appendix
\section{Glossary of Definitions}
\label{sec:glossary}

We collect some basic notations and definitions:

\begin{itemize}

\item 
\label{def:Prop}
$\Prop$ is the (impredicative) universe of propositions.

\vspace{0.2cm}
\item
\textit{Natural numbers:}
\label{def:nat}
{ \hfill 
  $ n : \nat ::= 0 \mid \mathsf{S}\; n$
}

\vspace{0.2cm}
\item
\textit{Booleans:}
\label{def:bool}
{\hfill
$ b : \bool ::= \btrue \mid \bfalse $
}

\vspace{0.2cm}
\item
\textit{Unit type:}
{\hfill
	\label{def:unit}
	$\unit ::= \star$}

\vspace{0.2cm}
\item
\textit{Sum type:}
\label{def:sum}
{\hfill
$ \sumt X Y ::= \mathsf{inl} x \mid \mathsf{inr} y \quad (x : X, y : Y) $
}

\vspace{0.2cm}
\item
\textit{Option type:}
\label{def:option}
{\hfill $o : \option X ::= \None \mid \Some x \quad (x : X)$}

\vspace{0.2cm}
\item
\textit{Lists:}
\label{def:list}
{\hfill
  $ l : \List X ::= [\;] \mid x :: l \quad (x : X) $
}

\end{itemize}

\paragraph*{List operations}
We often rely on concatenation of of two lists $l_1 \app l_2$:
\label{def:app}
\begin{mathpar}{}
  [\,] \app l_2 := l_2 \and (x :: l_1) \app l_2 := x :: (l_1 \app l_2)
\end{mathpar}
\noindent
Also, we use an inductive predicate $\Forall_2\of{(X\to Y\to\Prop) \to \List X \to \List Y \to \Prop}$\label{def:Forall}\label{def:forall}
\begin{mathpar}
  \infer{~}{\Forall_2\,p\,[\,]\,[\,]} \and  
  \infer{p x y \and \Forall_2\,p\,l_1\,l_2}{\Forall_2\,p\,(x :: l_1)\,(y :: l_2)}
\end{mathpar}

\paragraph*{Characteristic relation}
\label{def:charrel}
The characteristic relation $\charrel p\of{X\to\bool\to\Prop}$ of a predicate $p\of{X\to\Prop}$ is introduced in \Cref{intro:charrel} as
$$\charrel p~:=~\lambda x b.
\begin{cases}
  px & \text{if } b = \btrue \\
  \neg p x & \text{if } b = \bfalse. \\
\end{cases}$$

\paragraph{Reducibility}
\label{def:redm}\label{def:redT}\label{def:redtt}
$\redm$ is many-one reducibility, introduced in \Cref{intro:redm}.
$\redtt$ is truth-table reducibility, introduced in \Cref{intro:redtt}.
$\redT$ is Turing reducibility, introduced in \Cref{intro:redT}.

\paragraph{Interrogations}
\label{def:interrogate}\label{def:subtree}
The interrogation predicate $\interrogate \sigma R \qs \ans$ is introduced in \Cref{intro:interrogate}.
It works on a tree $\sigma\of{\List A \to \sumt Q O}$.
We often also use trees taking an input, i.e.\ $\tau \of{I \to \List A \to \sumt Q O}$.
Given $\sigma$, we denote the subtree starting at path $l \of {\List A}$ with $\sigma \at l := \lambda l'.\;\sigma(l \app l')$.

\paragraph{Partial functions}
\label{def:part}\label{def:partial}\label{def:rho}
We use an abstract type of partial values over $X$, denoted as $\mathcal P X$,
with evaluation relation $\hasvalue\of{\mathcal P X \to X \to \Prop}$.
We set $\ofbox{X \pto Y} := \ofbox{X \to \mathcal P Y}$
and use
\begin{itemize}
\item  $\ret \of {X \pto X}$ with $\ret x \hasvalue x$,
\item $\bind \of { \mathcal P X \to (X \to \mathcal P Y) \to \mathcal P Y}  $ with $x \bind f \hasvalue y \leftrightarrow \exists v.\;x \hasvalue v \land f v \hasvalue y$,
\item $\mu \of {(\nat \to \mathcal P \bool) \to \mathcal P \nat} $ with $\mu f \hasvalue n \leftrightarrow f n \hasvalue \btrue \land \forall m < n.\;f m \hasvalue \bfalse$, and
\item $\partundef \of {\mathcal P X}$ with $\forall v.\;\partundef \not\hasvalue v$.
\end{itemize}

One can for instance implement $\mathcal P X$ as monotonic sequences $f:\nat \to \option X$, i.e.\ with $f n = \Some x \to \forall m \geq n.\;f m = \Some x$
and $f \hasvalue x := \exists n.\;f n = \Some x$.
For any implementation it is only crucial that the graph relation $\lambda x y.f x \hasvalue y$ for $f \of {\nat \pto \nat}$ is semi-decidable but cannot be proved decidable.
Semi-decidability induces a function $\rho\of{\mathcal{P} X \to \nat \to \option X}$, which we write as $\rho^n x$ with the properties that $x \hasvalue v \leftrightarrow \exists n.\;\rho^n x = \Some v$ and $\rho^n x = \Some v \to \forall m \geq n.\;\rho^m x = \Some v$.

\section{Extended Forms of Interrogations}
\label{sec:forms}

\subsection{Extended Interrogations with State}
\label{sec:extended}

As an auxiliary notion, before introducing the stalling interrogations, we first introduce extended interrogations with a state argument, but without stalling.
An extended tree is a function $\sigma : S \to \List A \pto \sumt {(S \times Q)} O$.
We define an inductive extended interrogation predicate $\interrogate \sigma R \qs \ans \mathbin{;} s \succ s'$ by:
\begin{mathpar}
	\infer{~}{\interrogate \sigma R {[]} {[]} \mathbin{;} s \succ s}
	
	\infer{\interrogate \sigma R \qs \ans \mathbin{;} s \succ s'' \and \sigma\;s''\;\ans \hasvalue \inl (s',q) \and R q a}
	{\interrogate \sigma R {\qs \app [q]} {\ans \app [a]} \mathbin{;} s \succ s'}
\end{mathpar}

A functional $F$
mapping $R \of{Q \to A \to \Prop}$
to a relation of type
$ \ofbox{I \to O \to \Prop}$ is computable via extended interrogations if
there are a type $S$, an element $s_0 : S$, and
a function $\tau \of{I \to S \to \List{A} \pto \sumt {(S \times Q)} O}$ such that
\[ \forall R\,i\,o.\;F R\,i\,o \leftrightarrow \exists \qs \; \ans \; s. ~ \interrogate {\tau i} R \qs \ans \mathbin{;} s_0 \succ s ~\land~ \tau\,i\,s\;\ans \hasvalue \inr o.  \]

Note that we do not pass the question history to the function here,
because if necessary it can be part of the type $S$.
\setCoqFilename{OracleComputability}
\begin{lemma}[][eOracleComputable_equiv]
	Computable functionals are computable via extended interrogations.
\end{lemma}
\begin{proof}
	Let $F$ be computable by $\tau$.
	Set $S$ to be any inhabited type with element $s_0$ and
	define \[\tau'\; i\; s\; l := \tau\,i\,l \bind \lambda x.\;
	\begin{cases}
	\ret(\inl(s, q)) & \text{if } x = \inl q  \\
	\ret o & \text{if } x = \inr o.
	\end{cases}.
	\]
	Then $\tau'$ computes $F$ via extended interrogations.
	\qed
\end{proof}

\begin{lemma}[][eOracleComputable_equiv]
	Functionals computable via extended interrogations are computable.
\end{lemma}
\begin{proof}
	Let $\tau \of{I \to S \to \List{A} \pto \sumt {(S \times Q)} O}$ compute $F$ via extended interrogations.
	Define $\tau' \of{S \to \List{A} \to I \to \List{A} \pto \sumt Q O}$ as
	\begin{mathpar}
		\tau'\;s\;l\;i\;[] := \tau\,i\,s\;l \bind \begin{cases}
			\ret (\inl q) & \text{if } x = \inl (e,q) \\
			\ret (\inr o) & \text{if } x = \inr o, \\
		\end{cases}
		\and
		\tau'\,s\,l\,i\,(a :: \ans) := \tau\,s\,l\,i \bind \lambda x.\;
		\begin{cases}
			\tau'\,s'\,(l \app [a])\,i\,\ans & \text{if } x = \inl (s',q) \\
			\ret (\inr o) & \text{if } x = \inr o. \\
		\end{cases}
	\end{mathpar}
	Then $\tau'\,s_0\,[]$ computes $F$.
	\qed
\end{proof}

\subsection{Stalling Interrogations}
\label{sec:stalling}
\setCoqFilename{TuringReducibility.OracleComputability}

We here give the left out proofs that stalling interrogations as described in \Cref{sec:PT} and interrogations are equivalent.

\begin{lemma}[][sOracleComputable_equiv]
	Functionals computable via extended interrogations are computable via stalling interrogations.
\end{lemma}
\begin{proof}
	Let $F$ be computable using a type $S$ and element $s_0$ by $\tau$ via extended interrogations.
	We use the same type $S$ and element $s_0$ and define $\tau'$ to never use stalling:
	\[\tau'\, i\, s\, l := \tau\,i\,s\,l \bind \lambda x.\;
	\begin{cases}
	\ret(\inl(s', \Some q)) & \text{if } x = \inl(s',q)  \\
	\ret (\inr o) & \text{if } x = \inr o.
	\end{cases}
	\]
	Then $\tau'$ computes $F$ via stalling interrogations.
	\qed
\end{proof}

\begin{lemma}[][sOracleComputable_equiv]
	Functionals computable via stalling interrogations are computable via extended interrogations.
\end{lemma}
\begin{proof}
	Take $\tau \of{I \to S \to \List{A} \pto \sumt {(S \times \option Q)} O}$ computing $F$ via stalling interrogations.
	We construct $\tau'\,i\,s\,\ans$
	to iterate the function $\lambda s'.\;\tau\,i\,s'\,\ans$ of type $S \pfun \sumt {(S \times \option Q)} O$.
	If $\ask(s'',\None)$ is returned, the iteration continues with $s''$.
	If $\ask (s, \Some q)$ is returned, $\tau'\,i\,s\ans$ returns $\ask (s, q)$.
	If $\out o$ is returned, $\tau'\,i\,s\,\ans$ returns $\out o$ as well.
	
	We omit the technical details how to implement this iteration process using unbounded search $\mu : (\nat \pfun \bool) \pfun \nat$.
	\qed
\end{proof}

\subsection{Proofs of Closure Properties}
\label{sec:proofs}

We here give the proofs that executing two computable functionals one after the other,
composing computable functionals,
and performing an unbounded search on a computable functional are all computable operations as stated in \Cref{sec:closure}.
We explain the tree constructions, which are always the core of the argument.
The verification of the trees are then tedious but relatively straightforward inductions, we refer to the Coq code for full detail.

\begin{proof}[of \Cref{coq:computable_bind}]
  Let $\tau_1$ compute $F_1$ maping relations $R \of{Q\to A \to \Prop}$ to relations of type $\ofbox{I \to O' \to \Prop}$,
  and $\tau_2$ compute
  $F_2$ mapping relations $R \of{Q\to A \to \Prop}$ to relations of type $\ofbox{(I \times O') \to O \to \Prop}$.

  To compute the functional mapping an oracle $R \of{Q\to A \to \Prop}$ to a computation $\lambda i o.\,\exists o' \of O'.\,F_1\,R\,i\,o' \land F_2\,R\,(i,o')\,o$
  of type $\ofbox{I \to O \to \Prop}$
  we construct a stalling tree with state type $\option{(O' \times \nat)}$ and starting state $\None$.
  The intuition is that the state $s$ remains $\None$ as long as $\tau_1$ asks questions,
  and once an output $o'$ is produced we save it and the number of questions that were asked until then in the state, which remains unchanged after.
  Then, $\tau_2$ can ask questions, but since $\ans$ contains also answers to questions of $\tau_1$, we drop the first $n$ before passing it to $\tau_2$.
  
  Formally, the tree takes as arguments the input $i$, state $s$ ans answer list $\ans$,
  and returns
  {
  \footnotesize
  \[
  \begin{cases}
    \ret (\inl (\None, \Some q)) & \text{if } s = \None, \tau_1\,i\,\ans \hasvalue \Some (\inl q) \\
    \ret (\inl (\Some (o', | \ans |), \None)) & \text{if } s = \None, \tau_1\,i\,\ans \hasvalue \Some (\inr o') \\
    \ret (\inl (\Some (o',n), \Some q)) & \text{if } s = \Some(o', n), \tau_2\,(i, o')\,(\ans \uparrow_n) \hasvalue \Some (\inl q) \\
    \ret (\inl (\Some (o',n), \Some q)) & \text{if } s = \Some(o', n), \tau_2\,(i, o')\,(\ans \uparrow_n) \hasvalue \Some (\inr o) 
  \end{cases}
  \]}
where $\ans \uparrow_n$ drops the first $n$ elements of $\ans$.
Note that formally, we use bind to analyse the values of $\tau_1$ and $\tau_2$, but just write a case analysis on paper. \qed
\end{proof}

\begin{proof}[of \Cref{coq:computable_comp}]
  Let $\tau_1$ compute $F_1$ mapping relations $R \of{Q\to A \to \Prop}$ to relations $\ofbox{X \to Y \to \Prop}$,
  and $\tau_1$ compute $F_2$ mapping relations $R \of{X\to Y \to \Prop}$ to relations $\ofbox{I \to O \to \Prop}$.
  We construct a stalling tree $\tau$ computing a functional mapping $R \of{Q\to A \to \Prop}$ to $\lambda i o.\;F_2\,(F_1 R)\,i\,o$
  of type $\ofbox{I \to O \to \Prop}$.

  Intuitively, we want to execute $\tau_2$.
  Whenever it asks a question $x$, we record it and execute $\tau_1\,x$ to produce an answer.
  Since the answer list $\ans$ at any point will also contain answers of the oracle produces for any earlier question $x'$ of $\tau_2$, we record furthermore how many questions were already asked to the oracle to compute $\tau_1 x$.
  
  As state type, we thus use $\List{(X\times Y)} \times \option{(X \times \nat)}$,
  where the first component remembers questions and answers for $\tau_2$,
  and the second component indicates whether we are currently executing $\tau_2$ (then it is $\None$), or $\tau_1$, when it is $\Some(x,n)$ to indicate that on answer list $\ans$ we need to run $\tau_1\,x\,(\ans\downarrow^n)$, where $\ans \downarrow^n$ contains the last $n$ elements of $\ans$.
  The initial state is $([\,], \None)$.

  We define $\tau$ to take as arguments an input $i$, a state $(t, z)$, and an answer list $\ans$ and return
  \[
    \begin{cases}
      \inr o & \text{if } x = \None, \tau_2\,i\,(\map\,\pi_2\,t) \triangleright \inr o \\
      \inl (t, \Some(x,0), \None) & \text{if } x = \None, \tau_2\,i\,(\map\,\pi_2\,t) \triangleright \inl x \\
      \inl (t, \Some (x, \succN n), \Some q) & \text{if } x = \Some(x,n), \tau_1\,x\,(\ans\uparrow^n) \triangleright \inl q \\
      \inl (t \app [(x,y)], \None, \None) & \text{if } x = \Some(x,n), \tau_1\,x\,(\ans\uparrow^n) \triangleright \inr y 
    \end{cases}
  \]

  Intuitively, when we are in the mode to execute $\tau_2$ and it returns an output, we return the output.
  If it returns a question $x$, we change mode and stall.
  When we are in the mode to execute $\tau_1$ to produce an answer for $x$, taking the last $n$ given answers into account and it asks a question $q$, we ask the question and indicate that now one more answer needs to be taken into account.
  If it returns an output $y$, we add the pair $[(x,y)]$ to the question answer list for $\tau_1$, change the mode back to execute $\tau_2$, and stall. \qed
\end{proof}

\begin{proof}[of \Cref{coq:computable_search}]
  We define a tree $\tau$ computing the functional mapping $R\of{(I \times \nat) \to\bool\to\Prop}$ to
  the following relation of type $\ofbox{I \to \nat \to \Prop}$:
  $\lambda i n.\; R\, (i,n)\, \btrue \land \forall m < n.\; R\, (i, m)\, \bfalse$.

  \[\tau\,i\,\ans :=
  \begin{cases}
    \ret (\inr i) & \text{if } \ans[i] = \btrue \\
    \ret (\inl (i, |\ans|)) & \text{if } \forall j.\, \ans[j] = \bfalse
  \end{cases}
  \]
  Note that a function $\mathsf{find}\,\ans$ computing the smallest $i$ such that $\ans$ at position $i$ is $\btrue$, and else returning $\None$ is easy to implement.

  Intuitively, we just ask all natural numbers as questions in order.
  On answer list $l$ with length $n$, this means we have asked $[0,\dots,n-1]$.
  We check whether for one of these the oracle returned $\btrue$, and else ask $n = |l|$.
  \qed
\end{proof}

\section{Relation to Bauer's Turing Reducibility}

\setCoqFilename{TuringReducibility.Bauer}

We show the equivalence of the modulus continuity as defined in \Cref{coq:cont_to_cont} with the order-theoretic characterisation used by Bauer~\cite{bauer2020Wisc}.
The latter notion is more sensible for functionals acting on functional relations, so we fix some
$$F:(Q \rightsquigarrow A)\to (I \rightsquigarrow O)$$
where $X\rightsquigarrow Y$ denotes the type of functional relations $X\to Y \to \Prop$.
To simplify proofs and notation, we assume extensionality in the form that we impose $R=R'$ for all $R,R':X \rightsquigarrow Y$ with $Rxy\leftrightarrow R'xy$ for all $x:X$ and $y:Y$.

To clarify potential confusion upfront, note that Bauer does not represent oracles on $\nat$ as (functional) relations but as pairs $(X,Y)$ of disjoint sets with $X,Y:\nat \to \Prop$, so his oracle computation operate on such pairs.
However, since such a pair $(X,Y)$ gives rise to a functional relation $R: \nat \rightsquigarrow \bool$ by setting $R\,n\,b := (X\,n \land b = \btrue )\lor (Y\,n \land b = \bfalse)$ and, conversely, $R: \nat \rightsquigarrow \bool$ induces a pair $(X,Y)$ via $X\,n := R\,n\,\btrue$ and $Y\,n := R\,n\,\bfalse$, Bauer's oracle functionals correspond to our specific case of functionals $\ofbox{(\nat\rightsquigarrow \bool)\to (\nat \rightsquigarrow \bool)}$.
He then describes the computable behaviour of an oracle functional by imposing continuity and a computational core operating on disjoint pairs $(X,Y)$ of enumerable sets that the original oracle functional factors through, which in our chosen approach correspond to the existence of computation trees.
So while the overall setup of our approach still fits to Bauer's suggestion, we now show that our notion of continuity is strictly stronger than his by showing the latter equivalent to modulus continuity.

Informally, Bauer's notion of continuity requires that $F$ preserves suprema, which given a non-empty directed set $:(Q\rightsquigarrow A)\to \Prop$ of functional relations requires that $F\,(\bigcup_{R\in S} R)=\bigcup_{R\in S}\, F\,R$, i.e.\ that the $F$ applied to the union of $S$ should be the union of $F$ applied to each $R$ in $S$.
Here directedness of $S$ means that for every $R_1,R_2\in S$ there is also $R_3\in S$ with $R_1,R_2\subseteq R_3$, which ensures that the functional relations included in $S$ are compatible so that the union of $S$ is again a functional relation.

\begin{lemma}[][modulus_continuous_to_Bauer_continuous]
	If $F$ is modulus-continuous, then it preserves suprema.
\end{lemma}

\begin{proof}
	First, we observe that $F$ is monotone, given that from $F\,R\,i\,o$ we obtain some modulus $L:\List Q$ that directly induces $F\,R'\,i\,o$ for every $R'$ with $R\subseteq R'$.
	
	So now $S$ be directed and non-empty, we show both inclusions separately.
	First $\bigcup_{R\in S}\, F\,R \subseteq F\,(\bigcup_{R\in S} R)$ follows directly from monotonicity, since if $F\,R\,i\,o$ for some $R\in S$ we also have $F\,(\bigcup_{R\in S} R)\,i\,o$ given $R\subseteq \bigcup_{R\in S} R$.
	
	Finally assuming $F\,(\bigcup_{R\in S} R)\,i\,o$, let $L:\List Q$ be a corresponding modulus, so in particular $L\subseteq \mathsf{dom}(\bigcup_{R\in S} R)$.
	Using directedness (and since $S$ is non-empty), by induction on $L$ we can find $R_L\in S$ such that already $L\subseteq \mathsf{dom}(R_L)$.
	But then also $F\,R_L\,i\,o$ since $L$ is a modulus and $R_L$ agrees with $\bigcup_{R\in S} R)$ on $L$.
	\qed
\end{proof}

\begin{lemma}[][Bauer_continuous_to_continuous]
	If $F$ is preserves suprema, then it is modulous continuous.
\end{lemma}

\begin{proof}
	Again, we first observe that $F$ is monotone, given that for $R\subseteq R'$ the (non-empty) set $S:=\{R,R'\}$ is directed and hence if $F\,R\,i\,o$ we obtain $F\,R'\,i\,o$ since $R'=\bigcup_{R\in S} R$.
	
	Now assuming $F\,R\,i\,o$ we want to find a corresponding modulus.
	Consider
	$$S:=\{R_L \mid L \subseteq \mathsf{dom}(R) \}$$
	where $R_L\,q\,a := q\in L \land R\,q\,a$, so $S$ contains all terminating finite subrelations of $R$.
	So by construction, we have $R=\bigcup_{R\in S} R$ and hence $F\,(\bigcup_{R\in S} R)\,i\,o$, thus since $F$ preserves suprema we obtain $L \subseteq \mathsf{dom}(R)$ such that already $F\,R_L\,i\,o$.
	The remaining part of $L$ being a modulus for $F\,R\,i\,o$ follows from monotonicity.
	\qed
\end{proof}

\end{document}